\documentclass[10pt,twocolumn,twoside] {IEEEtran}

\usepackage{amsmath,graphicx}
\usepackage{amsthm}
\usepackage{amscd}
\usepackage{amssymb}
\usepackage{latexsym}
\usepackage{epsfig}
\usepackage{graphicx}
\usepackage{caption}
\usepackage{subcaption}
\usepackage{cite}
\usepackage{multirow}
\usepackage{float}
\usepackage[english]{babel}
\usepackage{tikz}
\usetikzlibrary{patterns}

\newtheorem{theorem}{Theorem}[section]
\newtheorem{lemma}[theorem]{Lemma}
\newtheorem{prop}[theorem]{Proposition}

\newtheorem{dfn}[theorem]{Definition}
\newtheorem{remark}[theorem]{Remark}

\makeatletter
\newcommand*{\rom}[1]{\expandafter\@slowromancap\romannumeral #1@}
\makeatother
\def\a{\mathbf{a}}
\def\r{\mathbf{r}}
\def\b{\mathbf{b}}

\def\y{\mathbf{y}}
\def\A{\mathbf{A}}

\def\F{\mathbf{F}}

\def\G{\mathbf{G}}
\def\Q{\mathbf{Q}}

\def\g{\mathbf{g}}
\def\tg{\tilde{\mathbf{g}}}

\def\D{\mathbf{D}}
\def\one{\mathbf{1}}
\def\x{\mathbf{x}}
\def\z{\mathbf{z}}

\def\mset{\mathcal{A}}
\def\P{\mathbf{P}}
\def\Prob{\mathbb{P}}

\def\E{\mathbb{E}}
\def\Emtx{\mathbf{E}}

\def\C{\mathbb{C}}

\def\Cm{\mathbf{C}}
\def\W{\mathbf{W}}

\def\R{\mathbf{R}}

\def\u{\mathbf{u}}

\def\I{\mathbf{I}}
\def\bigo{\mathbf{O}}
\def\S{\mathbf{S}}

\def\U{\mathbf{U}}
\def\t{\mathbf{t}}

\def\V{\mathbf{V}}
\def\Vset{\mathcal{V}}
\def\v{\mathbf{v}}

\def\OOmega{\boldsymbol{\Omega}}
\def\SSigma{\boldsymbol{\Sigma}}
\def\ssigma{\boldsymbol{\sigma}}

\def\H{\mathbf{H}}

\def\eps{\boldsymbol{\epsilon}}
\def\bxi{\boldsymbol{\xi}}
\def\PPhi{\boldsymbol{\Phi}}
\def\PPsi{\boldsymbol{\Psi}}

\def\ex{\mathrm{e}}
\def\ric{\delta_s}
\def\sset{\mathcal{D}_{s,n}}

\def\B{\mathbf{B}}
\def\tB{\tilde{\mathbf{B}}}
\def\tn{\tilde{n}}
\def\hn{\hat{n}}
\def\zero{\mathbf{0}}
\def\xnorm{\|\x\|_2}

\def\ttwo{2\rightarrow2}

\def\diag{\text{diag}}

\def\aalpha{\boldsymbol{\alpha}}

\def\LLamda{\boldsymbol{\Lambda}}
\def\llamda{\boldsymbol{\lambda}}
\def\ggamma{\boldsymbol{\gamma}}

\def\Vx{\V_{\x}}
\def\Vy{\V_{\y}}
\def\tinf{\tilde{\infty}}
\def\bigo{\mathcal{O}}
\def\w{\mathbf{w}}

\begin{document}

\title{Modulated Unit-Norm Tight Frames for Compressed Sensing}

\author{Peng~Zhang, Lu~Gan, Sumei~Sun, and~Cong~Ling
\thanks{P. Zhang and C. Ling are with the Department
of Electrical and Electronic Engineering, Imperial College London, London,
SW7 2AZ, UK (e-mail: p.zhang12@imperial.ac.uk, cling@ieee.org) .}
\thanks{L. Gan is with the School of Engineering and Design, Brunel University, London, UB8 3PH, UK (e-mail: lu.gan@brunel.ac.uk).}
\thanks{S. Sun is with the Institute for Infocomm Research, A$^{*}$STAR, Singapore, 138632, Singapore (e-mail: sunsm@i2r.a-star.edu.sg).}}




\maketitle

\begin{abstract}
In this paper, we propose a compressed sensing (CS) framework that consists of three parts: a unit-norm tight frame (UTF), a random diagonal matrix and a column-wise orthonormal matrix. We prove that this structure satisfies the restricted isometry property (RIP) with high probability if the number of measurements $m=\bigo(s\log^2s\log^2n)$ for $s$-sparse signals of length $n$ and if the column-wise orthonormal matrix is bounded. Some existing structured sensing models can be studied under this framework, which then gives tighter bounds on the required number of measurements to satisfy the RIP. More importantly, we propose several structured sensing models by appealing to this unified framework, such as a general sensing model with arbitrary/determinisic subsamplers, a fast and efficient block compressed sensing scheme, and structured sensing matrices with deterministic phase modulations, all of which can lead to improvements on practical applications. In particular, one of the constructions is applied to simplify the transceiver design of CS-based channel estimation for orthogonal frequency division multiplexing (OFDM) systems.
\end{abstract}


\section{Introduction}
Compressed sensing (CS) as an emerging field has attracted vast consideration over recent years in the areas of applied mathematics, computer science, and electrical engineering \cite{EJC05robustuncertainty,EJC06nearoptimal,Donoho06compressedsensing,eldar2012compressed,foucart2013mathematical}. The theory provides an efficient way to solve an ill-conditioned linear inverse problem with the prior knowledge that the signal of interest is sparse or compressible. A length-$n$ signal $\aalpha$ is said to be $s$-sparse when it is in the form of an orthogonal signal representation. That is, $\aalpha$  can be decomposed as $\aalpha=\PPsi\x$, where the unitary matrix $\PPsi\in\C^{n\times n}$ is the sparsifying transform (or orthobasis), and $\x\in\C^n$ has $s$ non-zero entries, i.e. $\|\x\|_0:=|\{l:\x_l\neq 0\}|\leq s$. Similarly, a signal is $s$-compressible if its orthogonal representation $\x$ can be approximated by $s$ non-zero entries. The CS measurement model is expressed as
\begin{align}\label{eqn: cs measurement model}
    \y=\PPhi\aalpha+\w=\A\x+\w,
\end{align}
where $\PPhi\in\C^{m\times n}$, $m\leq n$, is referred to as the sensing matrix, $\y$ is the measurement vector, $\w$ is the noise vector and $\A$ is the product of the sensing matrix $\PPhi$ and the sparsifying transform $\PPsi$, $\A=\PPhi\PPsi$. The restricted isometry property (RIP) as a sufficient condition implies uniform and stable recovery of all $s$-sparse vectors via nonlinear optimization (e.g. $l_1$-minimization). Theoretically, there are two essential parameters measuring the performance of a CS setup: the probability of perfect (or nearly perfect) recovery of the unknown sparse vectors, and the corresponding requirement on the number of measurements. In this context, sensing matrices constructed from independent Gaussian/Bernoulli distributions are optimal in the sense that they can cope with any orthobasis such that the resultant matrix $\A$ satisfies the RIP with high probability if $m\geq\bigo(s\log n/s)$ \cite{Donoho06compressedsensing,EJC05robustuncertainty}.

Since Gaussian/Bernoulli random matrices incur large computation and storage costs in practical implementations, a wide variety of structured sensing matrices have been proposed in recent years. Despite considerable progress in the field, there are still some open questions in the aspect of theoretical analysis or practical implementations.
\begin{enumerate}
  \item In many numerical simulations, some structured sensing models, e.g. random demodulation \cite{tropp2010beyond}, exhibit comparable recovery performances to that of Gaussian/Bernoulli random matrices. However, there is still a large gap between their existing theoretical bounds on the number of measurements to satisfy the RIP and the optimal bound given by Gaussian/Bernoulli random matrices. Is it possible to reduce this gap through new mathematical tools?
  \item Many structured sensing models that involve a random subsampling operator have been proposed in literature, such as randomly subsampled orthogonal transforms \cite{EJC06nearoptimal,rudelson2008sparse}, random convolution \cite{romberg2009compressive,li2013convolutional}, random subsampling of bounded orthonormal system \cite{rauhut2010compressive} and etc. However, due to the constraints of practical implementations, measurement models that consist arbitrary/deterministic selection of each measurement vector are more preferable (e.g. radio interferometry and magnetic resonance imaging \cite{puy2012universal}). Can we design a general sampling scheme with arbitrary/deterministic subsampling operators?
  \item Although there exist some structured sensing models that have an arbitrary/deterministic subsampling operator, they perform well for sparse signals in very specific sparsifying bases, e.g. partial random circulant matrices \cite{krahmer2014suprema} only exhibit good recovery performance for sparse signals on the spatial domain. How to improve their compatibility with different sparsifying bases without introducing extra randomness?
\end{enumerate}

In this paper, we propose a unified framework for structured sensing models, thus obtaining positive answers to the questions above. Specifically, our framework consists of three parts: a unit-norm tight frame (UTF), a random diagonal matrix and a column-wise orthonormal matrix. The RIP analysis on the unified framework provides tighter bounds on the number of measurements to satisfy RIP for some structured sensing models by noticing that each of these constructions is a special case of the framework.

More importantly, we demonstrate that several new structured sensing models can be constructed and analyzed in the framework, including a general sensing model with arbitrary/determinisic subsamplers, random block diagonal matrices that support fast computation and efficient storage, and structured sensing matrices with deterministic phase modulations. Comparing with existing sensing models, these new designs lead to better implementation schemes in practical applications, such as imaging system and channel estimation of orthogonal frequency division multiplexing (OFDM) systems.

\subsection{Organization of the Paper}
The remainder of the paper is organized as follows. In Section \ref{sec: review of structured sensing}, we review structured sensing models commonly discussed in the CS literature and introduce the motivation of this paper. In Section \ref{sec: main results}, we present the main theorem for the RIP analysis on the proposed framework and provides tighter bounds on the number of measurements to satisfy the RIP for some existing structured sensing matrices. We apply the main theorem to the construction of several new sensing models in Section \ref{sec: new applications}, including a new channel estimation scheme for OFDM systems by employing the idea of deterministic phase modulation. Simulation results are given in Section \ref{sec: simulations}, followed by conclusions in Section \ref{sec: conclusion}. We defer all proofs to the Appendices.

\subsection{Notations and Preliminaries}\label{sec: notations and preliminaries}
We give the notations and review some important notions in compressed sensing. For a vector $\a$, we denote by $a_i$, ($i\in[n]=\{0,...,n-1\}$), the $i$-th element of this vector. We represent a sequence of vectors by $\a_0,...,\a_{n-1}$ and a column vector with $q$ ones by $\one_q$. For a matrix $\A$, $\A_{jk}$ denotes the element on its $j$-th row and $k$-th column. The vector obtained by taking the $j$-th row ($k$-th column) of $\A$ is represented by $\A_{(j,:)}$ ($\A_{(:,k)}$). $\A_{(1:q)}$ denotes the submatrix consisting of the first $q$ columns of $\A$. We also denote by $\A_0,...,\A_{n-1}$ a sequence of matrices. $\A^{-1}$ and $\A^*$ represent the inverse and the conjugate transpose of $\A$. $\A\otimes \B$ denotes the Kronecker product of $\A$ and $\B$. The Frobenius norm and the operator norm of matrix $\A$ are denoted by $\|\A\|_F=\sqrt{\text{tr}(A^*A)}$ and $\|\A\|_{2\rightarrow2}=\sup_{\xnorm=1}\|\A\x\|_2$ respectively.

We define $\F$ ($\F^*$) as the normalized (inverse) discrete Fourier transform (DFT) matrix the dimension of which will be clear from the context. For the identity matrix, we use the subscript to denote the dimension, i.e. $\I_l$ denotes the $l\times l$ identity matrix. We omit the subscript whenever the dimension is clear to see from the context.

We represent the block diagonal matrix generated from a set of $L$ matrices $\{\A_0, \A_1, \cdots, \A_{L-1}\}$, $\A_i\in\C^{p\times q}$ for $i\in[L]$, by
\begin{align*}
\diag([\A_i]_{i=0}^{L-1})=\begin{bmatrix}
                            \A_0 &  &  &  \\
                             & \A_1 &  &  \\
                             &  & \ddots &  \\
                             &  &  & \A_{L-1}
                          \end{bmatrix}.
\end{align*}

Let $\OOmega\subset\{0,\cdots,n-1\}$ be an arbitrary/deterministic set of cardinality $m$, and denote by $\R_{\OOmega}: \C^{n}\rightarrow\C^{m}$ the subsampling operator that restricts a vector $\x\in\C^n$ to its entries in $\OOmega$. Similarly, $\R_{\OOmega'}: \C^{n}\rightarrow\C^{m}$ represents a random subsampling operator where each elements in $\OOmega'$ is selected independently and uniformly from $\{0,\cdots,n-1\}$. We write $A\lesssim B$ if there is an absolute constant $c$ such that $A\leq cB$.
\subsubsection{Coherence parameter}
The coherence parameter $\mu(\A)$ of an $\tn\times n$ matrix $\A$ describes the maximum magnitude of the elements of $\A$ \cite{candes2011probabilistic}
\begin{align*}
\mu(\A)=\max_{\substack{0\leq j<\tn \\ 0\leq k<n}} |\A_{jk}|.
\end{align*}
For a unitary matrix $\PPsi\in\C^{n\times n}$, we have $\frac{1}{\sqrt{n}}\leq \mu(\PPsi)\leq 1$.
\subsubsection{Restricted Isometry Property}
One important notion that has been successfully used to establish \emph{uniform recovery guarantees} is the restricted isometry property (RIP). For a CS measurement model \eqref{eqn: cs measurement model}, uniform and stable recovery of all $s$-sparse signals in the sparsifying basis $\PPsi$ via nonlinear optimization (e.g. $l_1$-minimization) is ensured provided that the matrix $\A=\PPhi\PPsi$ satisfies the RIP. Therefore, in this paper, the product of the sensing matrix $\PPhi$ and the sparsifying transform $\PPsi$, e.g. $\A=\PPhi\PPsi$, is referred to as the sensing model.
\begin{dfn}[\cite{candes2005decoding}]
A matrix $\A\in\C^{m\times n}$ is said to satisfy the RIP of order $s$ and level $\delta$ if
\begin{align}\label{eqn: rip}
(1-\delta)\|\x\|_2^2\leq\|\A\x\|_2^2\leq(1+\delta)\|\x\|_2^2,
\end{align}
holds for all $s$-sparse vectors $\x\in\C^{n}$ with respect to the identity basis $\I$. The smallest $\delta$ that satisfies \eqref{eqn: rip} is called the restricted isometry constant of order $s$ and denoted by $\ric$.
\end{dfn}

\section{Review of structured sensing models}\label{sec: review of structured sensing}
Although Gaussian and Bernoulli random matrices have been shown to satisfy the RIP with optimal bound on the number of measurements, they have limitations in practice for several reasons: the design of the measurement matrix is usually subject to constraints of the application; the large computation and storage cost by using Gaussian and Bernoulli random matrices impedes their application in large scale problems. This leads to the study of structured sensing models. In this section, we briefly review standard structured sensing models in the CS literature.

\subsection{Randomly subsampled orthogonal system}
The first structured sensing model proposed in the literature consists of randomly chosen rows of the discrete Fourier matrix, e.g. $\A=\sqrt{\frac{n}{m}}\R_{\OOmega'}\F$ ($\A=\PPhi$, $\PPsi=\I$) and $\F$ here represents an $n\times n$ normalized DFT matrix \cite{EJC06nearoptimal}. This model known as \emph{random partial Fourier} can provide fast matrix multiplication by using fast Fourier transform (FFT) algorithm. However, it performs poorly when dealing with signals in other sparsifying basis, e.g. wavelet basis. To tackle this problem, a new group of sensing models have been proposed by adding a diagonal matrix to the random partial Fourier model. Their measurement matrices can be written as $\PPhi=\sqrt{\frac{n}{m}}\R_{\OOmega'}\F\diag(\t)$. \cite{gan2012golay,gan2013golay} show that the measurement matrices can efficiently sample a sparse signal in the identity or Fourier or wavelet basis when $\t$ is a Golay sequence. Whereas, \cite{do2012fast} demonstrates that the measurement matrices can guarantee faithful recovery for a sparse signal in any basis provided that $\t$ is a random Bernoulli vector.

Another type of structured sensing models arises in applications where convolutions are involved. The sensing model, known as \emph{random convolution} based CS, is formed by randomly selecting rows from a circulant matrix, e.g. $\A=\PPhi\PPsi=\frac{1}{\sqrt{m}}\R_{\OOmega'}\H_{\r}\PPsi$ and $\H_{\r}$ is a circulant matrix formed by a vector $\r$, i.e.
\begin{align*}
    \H_{\r}=\begin{bmatrix} \r_0 & \r_{n-1} & \cdots & \r_1\\
                            \r_1 & \r_0 & \cdots & \r_2 \\
                            \vdots & \vdots & \ddots & \vdots\\
                            \r_{n-1} & \r_{n-2} & \cdots & \r_0
                            \end{bmatrix}.
\end{align*}
For this setup, the vector $\r$ can be either random or deterministic: when each element of $\t$ is drawn from i.i.d. Gaussian/Bernoulli random variables, the model can ensure good recovery performance for sparse signals in any basis \cite{romberg2009compressive}; when $\t$ forms a nearly perfect sequence (e.g. Golay sequence), good recovery guarantee can be proved for signals that are sparse in identity or Fourier or DCT basis \cite{li2013convolutional}.

Actually, all of the above structured sensing models can be analyzed in a general framework that consists of randomly subsampling orthogonal system \cite{rauhut2010compressive}. Suppose $\B$ is an arbitrary unitary matrix, it has been proved that
\begin{align*}
\A=\sqrt{\frac{n}{m}}\R_{\OOmega'}\B
\end{align*}
satisfies the RIP with high probability if $m\geq c\delta^{-2}sn^2\mu^2(\B)\log^4n$. Besides those have been reviewed so far, this framework encompasses many other structured sensing models that consist of randomly subsampling operators including \cite{rudelson2008sparse,baraniuk2008single,ma2009single,wiaux2010spread,puy2012universal}. In \cite{candes2011probabilistic}, a more general structure was proposed and analyzed based on \emph{nonuniform recovery guarantees}, e.g. no RIP is shown. We note that structures considered in both \cite{rauhut2010compressive} and \cite{candes2011probabilistic} consist in selecting each row vector independently from the others. However, as will be introduced in the following subsections, there exist other structured sensing models that can not be grouped into this category.

\subsection{System with fixed sampling locations}\label{sec: sys with fixed sampling}
The \emph{partial random circulant} sensing model can be expressed as $\A=\PPhi\PPsi=\frac{1}{\sqrt{m}}\R_{\OOmega}\H_{\r}$, where $\PPsi=\I$ and $\H_{\r}$ represents a circulant matrix formed by a random vector $\r$ \cite{krahmer2014suprema}. It is different from the random convolution sensing model for two reasons: first, it consists of an arbitrary/deterministic subsampling operator instead of a random one; second, it only copes with sparse signals in the identity basis ($\PPsi=\I$).

The second structured sensing model, named as \emph{random demodulation}, is motivated by analog to digital conversion. Let $\one_q$ represent a column vector with $q$ ones, the model can be represented as \cite{tropp2010beyond}
\begin{align}\label{eqn: random demodulation}
\A=\P_1\SSigma\tilde{\F}
\end{align}
where
\begin{align*}
\P_1&=\I_m\otimes\one_q^T\\
    &=\begin{bmatrix}\text{$1$ $1$ $\cdots$ $1$} & & & \\
                                           & \text{$1$ $1$ $\cdots$ $1$} & & \\
                                           & & \cdots & \\
                                           & & & \text{$1$ $1$ $\cdots$ $1$}\end{bmatrix},
\end{align*}
and $\SSigma=\diag(\ssigma)$ with $\ssigma$ being a length-$n$ Rademacher vector ($n=mq$). $\tilde{\F}\in\C^{n\times n}$ denotes a permuted DFT matrix, i.e.,
\begin{align*}
\tilde{\F}_{jk}=\frac{1}{\sqrt{n}}\ex^{-2\pi i jk/n},
\end{align*}
where $j=0,\cdots,n-1$ and $k=0,\pm1,\cdots,\pm(\frac{n}{2}-1),\frac{n}{2}$. In random demodulation, the matrix $\P_1$ is known as the integrator. Here, $\PPhi=\P_1\SSigma$ and $\PPsi=\tilde{\F}$.

\subsection{Multiple channel systems}\label{sec: mulitple channel sys}
This type of sensing models is constructed by concatenating structured matrices. There are mainly two structured sensing models belonging to this type: \emph{random probing} and \emph{compressive multiplexing}. The random probing model was proposed to estimate the channel response between multiple source-receiver pairs, which can be applied in seismic exploration, channel estimation of MIMO systems and coded aperture imaging \cite{romberg2010sparse}. Let $\G_i=\diag(\g_i)$ with $\g_i\in\C^{m}$ being the random probe signals. The random probing model can be represented as below.
\begin{align}\label{eqn: random probing}
\A=\F^*\begin{bmatrix}\G_0\F_{(1:q)} & \G_1\F_{(1:q)} & \cdots & \G_{L-1}\F_{(1:q)}\end{bmatrix}
\end{align}
where $\F$ represents an $m\times m$ normalized DFT matrix. In this model, $\A=\PPhi$ and $\PPsi=\I$. It is noted that each block $\F^*\G_i\F_{(1:q)}$ is a submatrix obtained by selecting the first $q$ columns of an circulant matrix. In \cite{slavinsky2011compressive}, the compressive multiplexing sensing model was proposed and applied in recovering of signals that are jointly sparse over the combined bandwidth of a number of spectrum channels. Mathematically, it can be represented as
\begin{align}\label{eqn: compressive multiplexing}
\A=\begin{bmatrix}\SSigma_0 & \SSigma_1 & \cdots & \SSigma_{L-1}\end{bmatrix}\F,
\end{align}
where $\F$ is an $m\times m$ normalized DFT matrix, $\SSigma_i=\diag(\ssigma_i)$ and $\{\ssigma_i\}$ are independent length-$m$ Rademacher vectors. Here, $\PPhi=\begin{bmatrix}\SSigma_0 & \SSigma_1 & \cdots & \SSigma_{L-1}\end{bmatrix}$ and $\PPsi=\F$.

It is noted that the sensing models in this category can not be analyzed by the existing framework proposed in \cite{rauhut2010compressive,candes2011probabilistic} since none of these models consists in selecting each row vector independently from the others. Is it possible to find a new framework that encompasses these sensing models? At the first glance, the answer may be pessimistic since these four sensing models seem isolated to each other.

However, we will develop a unified framework and demonstrate that it includes all of the sensing models in Section \ref{sec: sys with fixed sampling} and \ref{sec: mulitple channel sys}. Generally, our framework and the one proposed in \cite{rauhut2010compressive,candes2011probabilistic} complement each other; many of the structured sensing models commonly discussed in CS can now be classified and analyzed in one of both frameworks. The contributions of our proposed framework are twofold. Firstly, it proves tighter RIP bounds on the required number of measurements for some existing structured sensing models (see Section \ref{sec: main results}). Secondly, our newly designed sensing models can bring various improvements in practical applications (see Section \ref{sec: new applications}).

\section{Main results}\label{sec: main results}
In this section, we present our main theoretical results on the recovery of sparse (or compressible) signals from structured measurements and demonstrate how to obtain tighter RIP bounds for some of the existing structured sensing models by using the proposed framework.

Before continuing, we pause to review the definition and useful properties of unit-norm tight frames (UTF) that are essential for our theorem. For more details, see \cite{christensen2002introduction}, for example.
\subsection{Unit-norm Tight Frames}
A set of vectors $\Vset=\{\v_i\}_{i\in[n]}$ in a complex Hilbert Space $\C^m$ is called a finite frame if
\begin{align*}
\alpha\xnorm^2\leq\sum_{i\in[n]}|\langle\v_i,\x\rangle|^2\leq\beta\xnorm^2,
\end{align*}
for all $\x\in\C^m$. If $\alpha=\beta$, then the frame is tight. When the frame vectors all have unit norm, i.e. $\|\v_i\|_2=1$, it is called a unit-norm frame. A unit-norm tight frame (UTF) has
\begin{align}\label{eqn: UTF frame bound}
\alpha=\frac{n}{m}.
\end{align}
We form an associated $m\times n$ matrix with the frame vectors as its columns
\begin{align*}
\V=\begin{bmatrix}
       \v_0 & \v_1 & \v_2 & \cdots & \v_{n-1}
     \end{bmatrix},
\end{align*}
then the following proposition can be adapted from \cite{tropp2005designing}.
\begin{prop}[Proposition 1 \cite{tropp2005designing}]\label{prop: UTF}
An $m\times n$ normalized matrix $\V$ is a UTF if and only if it satisfies one (hence both) of the following conditions.
\begin{itemize}
  \item The $m$ nonzero singular values of $\V$ equal $\sqrt{\frac{n}{m}}$.
  \item The rows of $\sqrt{\frac{m}{n}}\V$ form an orthonormal family.
\end{itemize}
\end{prop}
By Proposition \ref{prop: UTF}, it is easy to verify that the following matrices are UTFs.
\begin{align}
\P_1&=\I_m\otimes\one_q^T\nonumber\\
    &=\begin{bmatrix}\text{$1$ $1$ $\cdots$ $1$} & & & \\
                                           & \text{$1$ $1$ $\cdots$ $1$} & & \\
                                           & & \cdots & \\
                                           & & & \text{$1$ $1$ $\cdots$ $1$}\end{bmatrix},\label{eqn: p1}\\
\P_2&=\one_L^T\otimes\F^*=\begin{bmatrix}\F^* & \F^* & \cdots & \F^* \end{bmatrix},\label{eqn: p2}\\
\P_3&=\one_L^T\otimes\I_m=\begin{bmatrix}\I & \I & \cdots & \I \end{bmatrix}\label{eqn: p3}.
\end{align}
In general, a UTF can be obtained from Harmonic frames or Gabor frames \cite{casazza2003equal}.

\subsection{Main theorem}
We are now ready to present the main theorem of this paper.
\begin{theorem}\label{thm: the main theorem}
Consider a framework that consists of three parts $\A=\U\D\tB$, where $\U\in\C^{m\times \tn}$ is a UTF, $\D=\diag(\bxi)$ is a diagonal matrix with $\bxi$ being a length-$\tn$ random vector with independent, zero-mean, unit-variance, and $r$-subgaussian entries, and $\tB\in\C^{\tn\times n}$ represents a column-wise orthonormal matrix, i.e. $\tB^*\tB=\I$. If, for $\delta\in (0,1)$,
\begin{align*}
m\geq c_1\delta^{-2}s\tn\mu^2(\tB)(\log^2s\log^2\hn),
\end{align*}
where $\hn :=\max\{\tn,n\}$ and $c_1>0$ is a constant, then with probability at least $1-\hn^{-(\log \hn)(\log s)^2}$, the restricted isometry constant of $\A=\U\D\tB$ satisfies $\ric\leq\delta$.
\end{theorem}

\begin{proof}Details of the proof are given in Appendix \ref{appndix: proof of main}.\end{proof}

In this paper, we coin the combination $\U\D$ a randomly modulated UTF since the diagonal of $\D$ is a random sequence. Clearly, the theorem still holds if $\tB$ is a unitary matrix. When $\tB$ is a bounded column-wise orthonormal matrix, i.e. $\mu(\tB)=\bigo(1/\sqrt{\tn})$, and $n=c_2\tn$ for a constant $c_2>0$, the bound on the number of measurements can be reduced to
\begin{align}
m\geq c_3\delta^{-2}s(\log^2s\log^2n), \quad c_3>0, \label{eqn: main theorem eqn 2}
\end{align}
which indicates that the number of measurements is linear in the sparsity level $s$ and (poly-)logarithmic in the signal dimension $n$. We term this construction a UDB (UTF, Diagonal and Bounded) framework.

The construction of sensing matrices by the use of UTF has been considered recently in literature \cite{bandeira2013road,chen2013projection,Tsiligianni2014construction}. However, none of their recovery performances is based on the RIP analysis. We refer the readers to \cite{vershynin2010introduction} for the background on $r$-subgaussian random variables/vectors. A simple example is a Rademacher or Steinhaus vector.

Our framework is both simple and general: first, it characterizes a variety of existing structured sensing models; second, many new structured sensing models can be constructed and analyzed within this framework (see Section \ref{sec: new applications}).

In the following subsection, we demonstrate how to obtain tighter RIP bounds for some structured sensing models by using the proposed framework. (See Table \ref{table rip bounds} for a summary of the comparison results.)
\begin{table*}
  \centering
  \caption{The RIP bounds of some sensing models obtained by the UDB framework}\label{table rip bounds}
  \begin{minipage}{0.75\textwidth}
  \centering
  \begin{tabular}{|c|c|c|c|}
    \hline
    Sensing models & Our bound & Previous bound & Ref \\
    \hline
    \hline
    Random demodulation & $m\geq \bigo(\delta^{-2}s\log^2s\log^2n)$ & $m\geq \bigo(\delta^{-2}s\log^6n)$ & Theorem $16$ \cite{tropp2010beyond} \\
    \hline
    Random probing & $m\geq \bigo(\delta^{-2}s\log^2s\log^2n)$ & $m\geq \bigo(\delta^{-2}s\log^6n)$ & Theorem $3.3$ \cite{romberg2010sparse} \\
    \hline
    Compressive multiplexing & $m\geq \bigo(\delta^{-2}s\log^2s\log^2n)$ & $m\geq \bigo(\delta^{-2}s\log^4n)$  \footnote{an expectation bound} & Theorem 3.1 \cite{romberg2009multiple} \\
    \hline
\end{tabular}
\end{minipage}
\end{table*}

\subsection{Tighter RIP Bounds}
With the help of our framework, the RIP analysis for some structured sensing models can be easily accomplished by noticing that the sensing model can be decomposed into three parts, all of which match exactly with those specified in Theorem \ref{thm: the main theorem}.

Firstly, it can be seen that the random demodulation sensing model \eqref{eqn: random demodulation} consists of three matrices, each of which matches with the three parts specified in our framework: $\P_1=\I\otimes\one_q^T$ is a UTF \eqref{eqn: p1}, $\tilde{\F}$ is a column-wise orthonormal matrix. Since $\mu(\tilde{\F})=1/\sqrt{n}$, the required number of measurements for this model to satisfy the RIP is given by \eqref{eqn: main theorem eqn 2}.

Secondly, the random probing sensing model \eqref{eqn: random probing} can be decomposed into
\begin{align}
\A=\P_2\diag(\tg)\Q \label{eqn: our random probing},
\end{align}
where $\P_2=\one_L^T\otimes\F^*$ is row-wise concatenation of $L$ inverse DFT matrices, $\tg:=\begin{bmatrix}\g_0^T & \g_1^T & \cdots & \g_{L-1}^T \end{bmatrix}$ and $\Q=\I_L\otimes\F_{(1:q)}\in\C^{\tn\times n}$ ($\tn=mL\geq n=qL$, $L=\bigo(1)$). Here, $\P_2$ is a UTF \eqref{eqn: p2}, and $\Q$ is a bounded column-wise orthonormal matrix with $\mu(\Q)=\frac{1}{\sqrt{m}}=\sqrt{\frac{L}{\tn}}$. Suppose each $\g_i$ is an independently subgaussian random vectors, an application of Theorem \ref{thm: the main theorem} leads to a better bound on the number of measurements than the existing results.

Similarly, the compressive multiplexing sensing model \eqref{eqn: compressive multiplexing} can be written as
\begin{align}\label{eqn: general compressive mux}
\A=\P_3\diag(\tilde{\ssigma})\B,
\end{align}
where $\P_3=\one_L^T\otimes\I_m\in\C^{m\times n}$ ($n=mL$, $L=\bigo(1)$) is row-wise concatenation of $L$ identity matrices, $\tilde{\ssigma}:=\begin{bmatrix}\ssigma_0^T & \ssigma_1^T & \cdots & \ssigma_{L-1}^T \end{bmatrix}$ with $\ssigma_i\in\C^m$ being independent Rademacher vectors and $\B=\I_L\otimes \F$ with $\F$ being an $m\times m$ normalized DFT matrix. Here, $\P_3$ is a UTF \eqref{eqn: p3}.

Besides achieving tighter RIP bounds, another benefit of analyzing these models in our framework is that these bounds still holds when the third decomposed part (the unitary matrix or column-wise orthonormal matrix) is replaced by any bounded unitary matrix (or column-wise orthonormal matrix). In this way, the above sensing models can be generalized and applied in more applications. For example, consider an image that is sparse in a basis $\PPsi\in\C^{n\times n}$ and $\PPsi$ is bounded unitary, then an imaging system by the sensing model \eqref{eqn: general compressive mux} with $\B=\PPsi$ first divides the $n$-pixel image into $L$ subimages, each of which is then randomly modulated by a Rademacher vector before combing onto a single detector of $m$ pixels. Similarly, we can compress $L$ $m$-pixel images ($n=mL$ pixels in total) into one image provided that each image is sparse on a bounded bases $\PPsi_i\in\C^{m\times m}$, $i=0,...,L-1$.

\section{Design of new sensing models}\label{sec: new applications}
In this section, we apply the general framework of Section \ref{sec: main results} to construct new structured sensing models and draw comparisons to existing literature where relevant.

We first show that random subsamplers in many existing structured sensing models can be replaced by arbitrary/deterministic subsamplers by noticing that any partial Fourier matrix is a UTF. This idea is then extended to a construction of fast and efficient random block diagonal matrices (Section \ref{sec: random to deterministic sub}).

Suppose $\B_1, \B_2, ..., \B_l$ is a set of arbitrary unitary matrix, then we can easily obtain the RIP analysis on $\A=\U\D\B_1\B_2\cdots\B_l$ by noticing that the product of any unitary matrices is still unitary. We construct the other two sensing models based on this observation: in Section \ref{sec: Convolutional CS with Deterministic Phase}, we demonstrate that the combination of deterministic phase modulations with partial random circulant matrices brings the new sensing matrices the compatibility with more sparsifying bases, and hence more practical applications; in the last part, we propose another sensing model and discuss a natural application of this model for the channel estimation of OFDM systems. This scheme can supersede previous CS based methods due to its capability to achieve a low Peak-to-Average Power Ratio (PAPR) and a low sampling rate simultaneously.

We note that construction of new structured sensing models based on the proposed framework is not limited to those included in this section. Our setup provides a simple and general design mechanism for structured sensing models due to the existence of multiple ways on constructing a UTF and a column-wise orthonormal matrix. Design of more structured sensing models by our framework is an interesting possible future direction.

\subsection{(Block) CS with arbitrary/deterministic subsamplers}\label{sec: random to deterministic sub}
Consider the following sensing model
\begin{align}\label{eqn: deterministic subsampler}
\A=\sqrt{\frac{n}{m}}\R_{\OOmega}\Emtx\D\PPsi,
\end{align}
where $\Emtx\in\C^{n\times n}$ is a normalized DFT or Hadamard matrix and $\PPsi\in\C^{n\times n}$ is an arbitrary unitary matrix. By Proposition \ref{prop: UTF}, it can be seen that $\U=\sqrt{\frac{n}{m}}\R_{\OOmega}\Emtx$ is a UTF. Then, Theorem \ref{thm: the main theorem} implies that this sensing model satisfies the RIP with high probability if
\begin{align}\label{eqn: arbitrary subsampling}
m\geq c\delta^{-2}sn\mu^2(\PPsi)\log^2s\log^2n.
\end{align}

As we mentioned in Section \ref{sec: review of structured sensing}, it has been shown that $\A=\sqrt{\frac{n}{m}}\R_{\OOmega'}\PPsi$ satisfies the RIP with high probability if $m\geq c\delta^{-2}sn\mu^2(\PPsi)\log^4n$ \cite{rauhut2010compressive}, which indicates similar requirement on $m$ as \eqref{eqn: arbitrary subsampling}. Hence, for many existing structured sensing models \cite{EJC06nearoptimal,rudelson2008sparse,romberg2009compressive,li2013convolutional,puy2012universal} encompassed in the framework \cite{rauhut2010compressive}, a replacement of the random subsampling operators $\R_{\OOmega'}$ by $\R_{\OOmega}\Emtx\D$ does not change the recovery performance of the new sensing model. Moreover, the arbitrary/deterministic subsampling operator can bring the new model advantages in practical implementations. For example, it is preferable to consider non-random measurements in the Fourier plane in realistic data acquisitions such as radio interferometry \cite{wiaux2009spread} and magnetic resonance imaging (MRI) \cite{wiaux2010spread,puy2012spread}.

In a similar way, we can construct random block diagonal matrices which support fast matrix multiplication. The measurement matrix is as below.
\begin{align*}
\tilde{\PPhi}&=\diag([\tilde{\PPhi}_i]_{i=0}^{L-1})=\begin{bmatrix}
                            \tilde{\PPhi}_0 &  &  &  \\
                             & \tilde{\PPhi}_1 &  &  \\
                             &  & \ddots &  \\
                             &  &  & \tilde{\PPhi}_{L-1}
                          \end{bmatrix}, \\
\tilde{\PPhi}_i&=\sqrt{\frac{q}{p}}\R_{\OOmega}\Emtx\diag(\bxi_i), \quad i=0,\cdots,L-1,
\end{align*}
where $\R_{\OOmega}: \C^{q}\rightarrow\C^{p}$ is an arbitrary/deterministic subsampling operator, $\Emtx$ is a $q\times q$ normalized DFT or Hadamard matrix and $\{\bxi_i\}$ are independent length-$q$ sub-Gaussian random vectors. We can easily prove the RIP of $\tilde{\A}=\tilde{\PPhi}\PPsi$ in our framework by noticing that
\begin{align*}
\tilde{\PPhi}&=\P_4\diag(\tilde{\bxi}),\\
\tilde{\bxi}&=\begin{bmatrix}\bxi_0^T & \bxi_1^T & \cdots & \bxi_{L-1}^T\end{bmatrix}^T,
\end{align*}
where $\P_4=\sqrt{\frac{q}{p}}\I_L\otimes(\R_{\OOmega}\F)$ is a UTF by Proposition \ref{prop: UTF}. Theorem \ref{thm: the main theorem} then indicates that $\tilde{\A}$ satisfies the RIP with high probability when
\begin{align*}
m\geq c_4\delta^{-2}sn\mu^2(\PPsi)\log^2s\log^2n,
\end{align*}
where $m=pL\leq n=qL$ and $L=\bigo(1)$. When the sparsifying basis $\PPsi$ is a bounded unitary matrix, our construction requires less memory and computations than existing random block diagonal matrices \cite{eftekhari2012restricted}.

\subsection{Convolutional CS with deterministic phase modulation}\label{sec: Convolutional CS with Deterministic Phase}
As we have reviewed in earlier section, the partial random circulant sensing model only provides good recovery guarantee for sparse signals in the identity basis. Here, we propose a new convolution-based CS scheme that not only retains the arbitrary/deterministic subsampling feature, but also exhibits compatibility with various sparsifying transforms without introducing additional randomness. Specifically, our scheme modulates the signal with a deterministic sequence prior to the partial random circulant matrix.

We denote $\LLamda=\diag(\llamda)$ a unitary diagonal matrix of size $n\times n$, i.e. $|\llamda_i|=1$ for all $0\leq i<n$. Our sensing matrix is as below
\begin{align*}
    \PPhi=\frac{1}{\sqrt{m}}\R_{\OOmega}\H_{\eps}\LLamda,
\end{align*}
where $\H_{\eps}$ denotes the circulant matrix generated from $\eps$, and $\R_{\OOmega}: \C^n \rightarrow\C^m$ represents an arbitrary/deterministic subsampling operator. Suppose $\eps=\F^*\bxi$, where $\bxi$ is a length-$n$ random vector with independent, zero-mean, unit-variance, and sub-Gaussian entries. Let $\D=\diag(\bxi)$, it follows that
\begin{align*}
    \PPhi=\sqrt{\frac{n}{m}}\R_{\OOmega}\F^*\D\F\LLamda.
\end{align*}
Consider a combined matrix $\A=\PPhi\PPsi$, where $\PPsi$ denotes the sparsifying basis, it can be observed that $\U=\sqrt{\frac{n}{m}}\R_{\OOmega}\F^*$ is a UTF and $\B=\F\LLamda\PPsi$ is a unitary matrix. By Theorem \ref{thm: the main theorem}, $\A$ satisfies the RIP with high probability if
\begin{align*}
m\geq c_4\delta^{-2}sn\mu^2(\F\LLamda\PPsi)\log^2s \log^2n.
\end{align*}
This result provides a generic framework for the RIP analysis on any sampling scheme that involves a partial random circulant matrix followed by a deterministic phase modulation and any orthornormal basis. For such sampling schemes, it implies that the recovery performance under nonlinear optimization such as $l_1$-minimization solely depends on the coherence parameter $\mu(\F\LLamda\PPsi)$.

In the case of no phase modulation, we can easily verify that the partial random circulant matrix provides a faithful recovery for sparse signals in the identity basis with the number of measurement $m\geq c\delta^{-2}s(\log s)^2(\log n)^2$ since the coherence parameter becomes $\mu(\F\I\I)=1/\sqrt{n}$. This conclusion coincides with that in \cite{krahmer2014suprema}. Similarly, the incompatibility of a partial random circulant matrix with sparse signals in the Fourier basis ($\PPsi=\F^*$) can be explained by Theorem \ref{thm: the main theorem}; the coherence parameter is $\mu(\F\I\F^*)=1$.

Our problem now is reduced to finding a proper modulation sequence such that in certain orthonormal bases, the corresponding coherence parameter is $\bigo(1/\sqrt{n})$. Next, we propose using Golay sequences for the phase modulations and demonstrate the performance of the corresponding sampling scheme for various orthonormal bases by analyzing the coherence parameter. To begin with, we briefly review the definition of Golay sequences.
\begin{dfn}[\cite{golay1961complementary}]
Consider two length-$n$ bipolar sequences $\a=[a_0,...,a_{n-1}]$, $\b=[b_0,...,b_{n-1}]$. Define two polynomials $A(z)=\sum_{k=0}^{n-1}\a_k z^k$ and $B(z)=\sum_{k=0}^{n-1}\b_k z^k$. $\a$ and $\b$ are said to be a Golay complementary pair if
\begin{align*}
|A(z)|^2+|B(z)|^2=2n
\end{align*}
for all $z$ on the unit circle, i.e. $|z|=1$.
\end{dfn}
This immediately gives us
\begin{align}
|A(z)|\leq\sqrt{2n}.\label{eqn: golay ploy bound}
\end{align}
Suppose that $\LLamda$ is a diagonal matrix whose diagonal entries form a Golay sequence. For $\PPsi=\I$, it can be easily shown that the coherence parameter gives $\mu(\F\LLamda\I)=1/\sqrt{n}$. In the case of $\PPsi^*$ is Fourier, DCT, block DCT, or Haar wavelet transform, the analysis on the coherence parameters has been studied in \cite{li2013convolutional,gan2013golay}.
\begin{lemma}[\cite{gan2013golay}, Lemma 1 and 2. \cite{li2013convolutional}, Corollary 1.]\label{lem: three coherences}
Denote by $\Cm$ the Type-II DCT transform, $\hat{\Cm}$ the block DCT transform and $\W$ the Haar wavelet transform.
For any Golay sequence $\llamda$, we have
\begin{align}
& \mu(\F\LLamda\F^*)\leq \sqrt{\frac{2}{n}} \nonumber\\
& \mu(\F\LLamda\Cm^*)\leq \frac{2}{\sqrt{n}} \nonumber\\
& \mu(\F\LLamda\hat{\Cm}^*)\leq \frac{2}{\sqrt{n}} \nonumber
\end{align}
If the Golay sequence are constructed by the Rudin-Shapiro iterative process \cite{golay1961complementary},
\begin{align}
& \mu(\F\LLamda\W^*)\leq \sqrt{\frac{2}{n}}. \label{eqn: golay wavelet}
\end{align}
\end{lemma}
The proof on \eqref{eqn: golay wavelet} was omitted in \cite{gan2013golay}. For completeness, we include the proof in Appendix \ref{app: coherence analysis}.
%

The additional Golay modulation process in our scheme can be easily implemented for the reason that a Golay sequence is simply a pseudorandom bipolar sequence. Actually, this process can be regarded as a phase modulation process, where only the two bipolar phases (i.e. +1, -1) are required. For the coded aperture imaging described in \cite{romberg2010sparse,ma2009single,rauhut2012restricted} a simple pre-modulation process by a Golay sequence extends the ability of the imaging system to handle images that are sparse in more bases. How to design the deterministic phase modulations such that the coherence parameter $\mu(\F\LLamda\PPsi)$ is small for other sparsifying bases $\PPsi$ is an interesting open problem.
\subsection{OFDM channel estimation with low speed ADC and low PAPR}\label{sec: ofdm channel estimation}
\begin{figure*}
\begin{center}
\begin{tikzpicture}
\draw [->, thick] (0,0.5) -- (0.5,0.5);
\node [left] at (0,0.5) {$\llamda_{n-1}$};
\draw [fill=black] (0.25,1) circle [radius=0.05];
\draw [fill=black] (0.25,1.5) circle [radius=0.05];
\draw [fill=black] (0.25,2) circle [radius=0.05];
\draw [->, thick] (0,2.5) -- (0.5,2.5);
\node [left] at (0,2.5) {$\llamda_{1}$};
\draw [->, thick] (0,3) -- (0.5,3);
\node [left] at (0,3) {$\llamda_{0}$};
\draw [fill=white, thick] (0.5,0) rectangle (1.5,3.5);
\node at (1,2) {IDFT};
\draw [->, thick] (1.5,0.5) -- (2,0.5);
\draw [fill=black] (1.75,1) circle [radius=0.05];
\draw [fill=black] (1.75,1.5) circle [radius=0.05];
\draw [fill=black] (1.75,2) circle [radius=0.05];
\draw [->, thick] (1.5,2.5) -- (2,2.5);
\draw [->, thick] (1.5,3) -- (2,3);
\draw [fill=white, thick] (2,0) rectangle (3,3.5);
\node at (2.5,2) {P/S};
\draw [->, thick] (3,2) -- (3.5,2);
\draw [fill=white, thick] (3.5,1.5) rectangle (4.5,2.5);
\node at (4,2.2) {DAC};
\node at (4,1.8) {$n$ Hz};
\draw [->, thick] (4.5,2) -- (5,2);
\draw [fill=white, thick] (5,1.75) rectangle (6.5,2.25);
\node at (5.75,2) {Channel};
\draw [->, thick] (6.5,2) -- (7.5,2);
\node [below]at (7,2) {$f(t)$};
\draw [fill=white, thick] (7.8,2) circle [radius=0.3];
\draw [thick] (7.59,1.79) -- (8.01,2.21);
\draw [thick] (7.59,2.21) -- (8.01,1.79);
\draw [->, thick] (8.1,2) -- (9.6,2);
\node [above] at (8.85,2) {$f(t)\cdot p(t)$};
\draw [->, thick] (7.8,2.8) -- (7.8,2.3);
\node [above] at (7.8,2.8) {$p(t)$};
\draw [fill=white, thick] (9.6,1.75) rectangle (11.5,2.25);
\node at (10.5,2) {Integrator};
\draw [->, thick] (11.5,2) -- (12.1,2);
\draw [fill=white, thick] (12.1,1.5) rectangle (13.1,2.5);
\node at (12.6,2.2) {ADC};
\node at (12.6,1.8) {$m$ Hz};
\draw [->, thick] (13.1,2) -- (13.6,2);
\node [right] at (13.6,2) {$\y$};
\end{tikzpicture}
\end{center}
\caption{Block Diagram for OFDM Channel Estimation}\label{fig: block diagram}
\end{figure*}
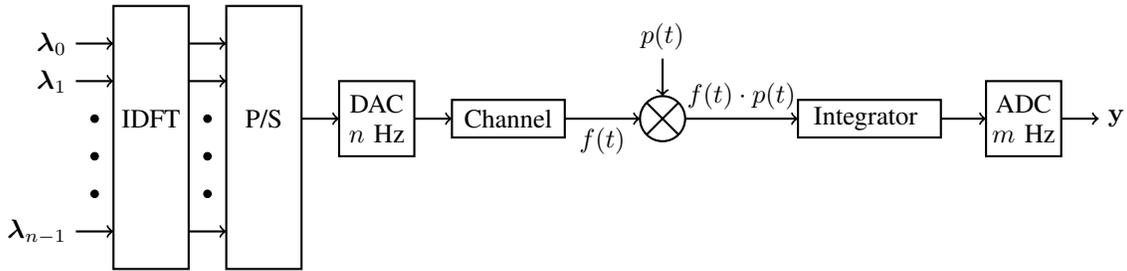
In this part, we present a novel CS-based OFDM channel estimation scheme. Our channel estimation scheme includes two key ingredients: a pilot signal generated by a Golay sequence and a random demodulator. Figure \ref{fig: block diagram} displays a block diagram for the scheme. We denote $\llamda$ a length-$n$ Golay sequence, which is employed as the pilot signal. The discrete signal is then passed through an $n$-point IDFT transform before being converted by a digital-to-analog converter (DAC) with a clock speed of $n$ Hz into an analog signal. At the receiver, the convolution of the transmitted signal and the channel response is sampled by a random demodulator. More specifically, the received signal is multiplied by a high-rate pseudonoise sequence, then integrated by a low-pass anti-aliasing filter (the integrator). The discrete samples are captured by a low rate ADC with a clock speed of $m$ Hz. We note that the cyclic prefix addition and removal block of the OFDM system are omitted in the block diagram.

Let $\x\in\C^{n}$ denote the channel response vector with $s$ taps, $\ssigma\in\C^n$ be the chipping sequence (e.g. a Rademacher vector) of the random demodulator and $\y\in\C^{m}$ represent the received samples. Then, the matrix form of our scheme can be expressed as
\begin{align*}
\y=\P_1\SSigma\F^*\LLamda\F\x+\w,
\end{align*}
where $\P_1=\I_m\otimes\one_q^T$, $n=mq$, $\SSigma=\diag(\ssigma)$, $\LLamda=\diag(\llamda)$ and $\w\in\C^n$ is the noise vector. Clearly, this model is in the UDB framework since $\P_1$ is a UTF. Due to the fact that $\mu(\F^*\LLamda\F)\leq\sqrt{\frac{2}{n}}$ (Lemma \ref{lem: three coherences}), this scheme guarantees stable channel estimation performance if
\begin{align*}
m\geq c_4\delta^{-2}s\log^2s \log^2n.
\end{align*}
When the channel is sparse, this result indicates that an $n$-resolution channel can be faithfully estimated by a low rate ADC (a clock speed of $m$ Hz).

\emph{Comparison with existing CS approaches:} In \cite{haupt2010toeplitz,berger2010sparse,meng2012compressive}, the pilot signals were generated by random sequences. Although only a low rate ADC is required at the receiver, the PAPR of the random sequences is asymptotically $\log n$ with probability $1$ \cite{sharif2004multicarrier}, which results in difficulty in the transmitter design in an OFDM system. In \cite{li2013convolutional}, the pilot signals were designed by Golay sequences, in which case the associated PAPR is bounded by 2. However, this scheme requires a random downsampling operator at the receiver, which cannot truly satisfies the requirement of low sampling frequencies due to the possibility of consecutive sampling. The combination of a Golay sequence and a random demodulator in our scheme resolves this dilemma; it achieves a low PAPR and a low sampling rate simultaneously. The only tradeoff in our scheme might be the implementation of the chipping sequence. However, we have seen end-to-end simulations of a transistor-level implementation \cite{laska2007theory} and practical circuit designs for the receivers with random demodulators \cite{Yoo2012subNyquist}. In \cite{Yoo2012subNyquist}, an effective instantaneous bandwidth of $2$ GHz is achieved with an aggregate digitization rate $f_s=320$ MSPS.
\begin{remark}[Random demodulation with deterministic phase modulation]
We note that the above structure can be regarded as a special case of the following sensing model
\begin{align*}
\A=\P_1\SSigma\F\LLamda\PPsi,
\end{align*}
where $\LLamda=\diag(\tilde{\llamda})$ is a unitary diagonal matrix of size $n\times n$, i.e. $|\tilde{\llamda}_i|=1$ for all $0\leq i<n$ and $\PPsi$ is an arbitrary unitary matrix. Clearly, this model is in the UDB framework, and it satisfies the RIP with high probability when
\begin{align*}
m\geq c_4\delta^{-2}sn\mu^2(\F\LLamda\PPsi)\log^2s \log^2n.
\end{align*}
\end{remark}

\section{Simulations}\label{sec: simulations}

In this section, we demonstrate the performance of the sensing models proposed in Section \ref{sec: Convolutional CS with Deterministic Phase} and \ref{sec: ofdm channel estimation}.

The first simulation demonstrates the improvement on the performance of partial random circulant matrices with the addition of Golay phase modulations. Figure \ref{fig: simulation1} shows the simulation results of compression and recovery on two different $256\times 256$ images based on existing and the proposed convolution-based CS models. We employ the sparsify averaging prior and the re-weighted BPDN from \cite{carrillo2012sparsity}. We set the input SNR as $30$ dB for both images and the down sampling ratio $r=\frac{1}{4}$. In the caption, R+R denotes the random convolution scheme consisting a random subsampling operator and a random sequence \cite{romberg2009compressive}, D+R represents the partial random circulant matrix constructed by an arbitrary/deterministic subsampling operator and a random sequence \cite{krahmer2014suprema}, R+E-Golay is the construction of a random subsampling operator and an extended Golay sequence \cite{li2013convolutional} and D+R+Golay-PM is our measurement scheme constructed by adding a Golay phase modulation to the partial random circulant matrix (D+R). The results show that the D+R scheme performs poorly on the recovery of the images. However, with a simple Golay phase modulation process, the proposed scheme exhibits comparable performance to the R+R and R+E-Golay schemes. Moreover, our scheme possesses the advantages of both the arbitrary/deterministic subsampling and the compatibility with varies sparsifying bases.

In the second simulation, we compare the performance of our proposed OFDM channel estimation scheme with those given by existing CS based methods. We set the number of carriers as $N=1024$ and collect $M=64$ samples at the receiver side. The channel model is the ATTC (Advanced Television Technology Center) and the Grande Alliance DTV laboratory ensemble E model. Here, the static case impulse response $\x(n)$ can be written as \cite{coleri2002channel}
\begin{align*}
\x&=\delta(n)+0.3162\delta(n-2)+0.1995\delta(n-17)\\
  &\quad +0.1296\delta(n-36)+0.1\delta(n-75)+0.1\delta(n-137).
\end{align*}
We vary the input SNR from $0$ dB to $30$ dB, and run $1000$ trials for each input SNR using the subspace pursuit algorithm \cite{dai2008subspace}. In Figure \ref{fig: simulation2}, the channel estimation performances based on three different schemes are shown. The `Random Phase' plot represents the performance by the channel estimation method proposed in \cite{meng2012compressive}, where the pilot signal is generated by a random vector and the received signal is sampled by a low-rate ADC. The `Golay+Rand subsampling' one indicates the performance by the method proposed in \cite{li2013convolutional}, where now the pilot signal is from a Golay sequence and the received signal needs to be randomly subsampled. The `Golay+Det subsampling' one is the performance given by our proposed scheme. It can be seen that all of these three schemes reveal similar reconstruction performance. However, our scheme achieves both a low PAPR and a low sampling rate simultaneously.
\begin{figure*}
\centering
  \centerline{\includegraphics[width=\textwidth]{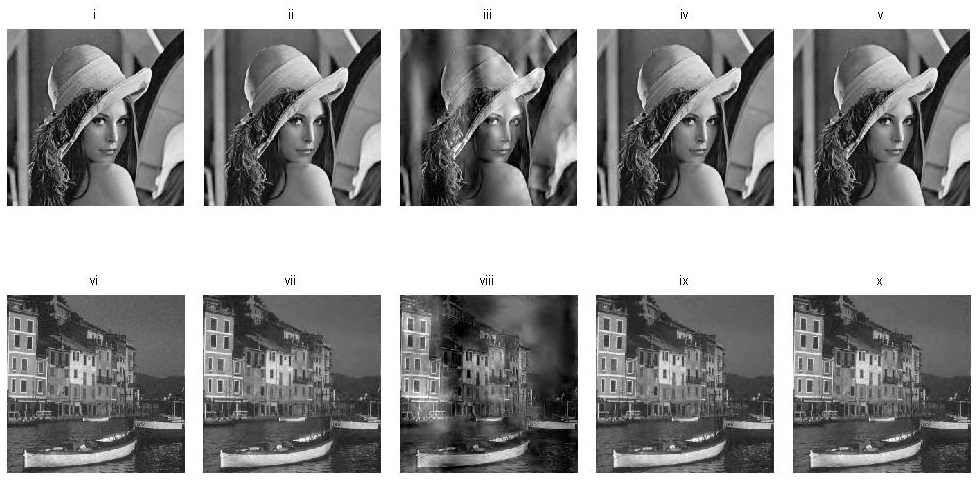}}
\caption{Partial Random Circulant Matrix with Golay Modulations. (i) Original image;(ii) R+R: $\text{SNR}=29.5703$ dB;(iii) D+R: $\text{SNR}=11.0224$ dB;(iv) R+E-Golay: $\text{SNR}=29.3712$ dB;(v) D+R+Golay-PM: $\text{SNR}=29.5606$ dB;(vi) Original image;(vii) R+R: $\text{SNR}=24.7431$ dB;(viii) D+R: $\text{SNR}=9.2057$ dB;(ix) R+E-Golay: $\text{SNR}=24.7586$ dB;(x) D+R+Golay-PM: $\text{SNR}=24.6195$ dB.}\label{fig: simulation1}
\end{figure*}
\begin{figure}
\includegraphics[width=0.5\textwidth]{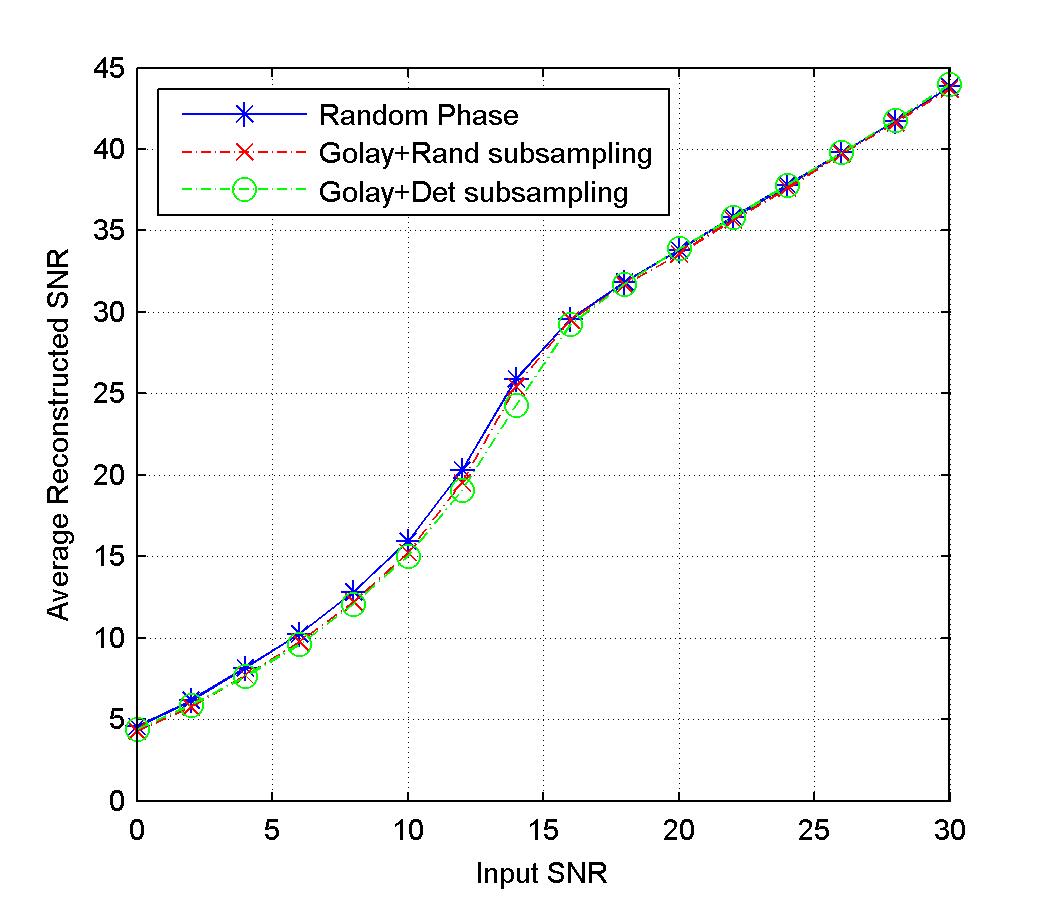}
\caption{Performance comparison of different CS-based OFDM channel estimation methods}\label{fig: simulation2}
\end{figure}

\section{Conclusion}\label{sec: conclusion}
In this paper, we proposed a generic CS framework for the construction of structured sensing models and proved its RIP based on the estimates of a suprema of chaos processes of a certain type. We have demonstrated that our framework is general and encompasses many existing/new structured sensing models. For any sensing model that involves selecting each measurement vector independently from each other, we provided a universal way to transform it into one with arbitrary/deterministic sampling operators. Moreover, we have proposed other structured sensing models that can motivate better practical implementation schemes, including distributed sensing, imaging and channel estimation. In particular, our OFDM channel estimation scheme outperforms existing CS-based methods by offering a low PAPR and a low sampling rate to the OFDM system simultaneously.
\appendices
\section{Proof of the Main Theorem}\label{appndix: proof of main}
In this section, we first review definitions and a useful lemma about covering number estimates, then present the proof.
\subsection{Covering number estimates}
A metric space is denoted by $(T,d)$, where $T$ is a set and $d$ is the notion of distance (metric) between elements of the set. For example, $(\mset,\|\cdot\|_{\ttwo})$ is a metric space where the matrices in the set $\mset$ have a distance measured by the operator norm, i.e. $d(\A_1,\A_2)=\|\A_1-\A_2\|_{\ttwo}$. For a metric space $(T,d)$, the covering number $N(T,d,u)$ is the minimal number of open balls of radius $u$ needed to cover $(T,d)$.

Define the set of all $s$-sparse signals with unit norm as
\begin{align}
\sset=\{\x\in\C^n: \|\x\|_2=1, \|\x\|_0\leq s\}.
\end{align}

The following lemma is summarized from the results in Section 8 of \cite{rauhut2010compressive}.
\begin{lemma}[\cite{rauhut2010compressive}]\label{lem: covering number estimates 0}
Let $\x_0, \x_1, \cdots, \x_{\tn-1}$ be vectors in $\C^n$ with $\|\x_i\|_{\infty}\leq K$ for $i\in[\tn]$. Consider the semi-norm
\begin{align}\label{eqn: semi norm}
\|\u\|_{X}:=\max_{i\in[\tn]}|\langle\x_i,\u\rangle|, \quad \u\in\C^n,
\end{align}
then we have the following two estimates on the covering number $N(\sset,\|\cdot\|_{X},t)$
\begin{align*}
& \sqrt{\log N(\sset,\|\cdot\|_{X},t)} \\
& \quad \leq 3K\sqrt{2s}\sqrt{\log(10\tn)\log(4n)}t^{-1}, \quad 0<t\leq 2K\sqrt{s},\\
& \sqrt{\log N(\sset,\|\cdot\|_{X},t)} \\
& \quad \leq \sqrt{2s}\sqrt{\log(\frac{en}{s})+\log(1+2K\sqrt{s}/t)}, \quad t>0.
\end{align*}
\end{lemma}
Next, we present a lemma that will be used in the following subsection.

Consider a matrix $\tB\in\C^{\tn\times n}$, and define the semi-norm for any vector $\u\in\C^n$ as
\begin{align*}
\|\u\|_{\tinf}:=\|\tB\u\|_{\infty}.
\end{align*}

By setting $\tB_{(i,:)}=\x_i$ for $i\in[\tn]$, it can be seen that $\|\u\|_{\tinf}=\|\u\|_{X}$ for any $\u\in\C^n$ and $\mu(\tB)=K$ as in the setting of Lemma \ref{lem: covering number estimates 0}. Therefore, the following lemma on the covering number $N(\sset,\|\cdot\|_{\tinf},t)$ is an immediate result of Lemma \ref{lem: covering number estimates 0}.
\begin{lemma}\label{lem: covering number estimates}
For a matrix $\tB\in\C^{\tn\times n}$ and an associated semi-norm $\|\cdot\|_{\tinf}$, we have
\begin{align*}
& \sqrt{\log N(\sset,\|\cdot\|_{\tinf},t)} \\
& \quad \lesssim \mu(\tB)\sqrt{s}\log\hn t^{-1}, \quad 0<t\leq 2\mu(\tB)\sqrt{s},\\
& \sqrt{\log N(\sset,\|\cdot\|_{\tinf},t)} \\
& \quad \leq \sqrt{2s}\sqrt{\log(\frac{en}{s})+\log(1+2\mu(\tB)\sqrt{s}/t)}, \quad t>0,
\end{align*}
where $\hn :=\max\{\tn,n\}$.
\end{lemma}
We note that the alternative forms of this lemma have been used in \cite{krahmer2014suprema}\cite{eftekhari2012restricted}.
\subsection{Proof of Theorem \ref{thm: the main theorem}}
\begin{proof}
By Proposition 2.5 of \cite{rauhut2010compressive}, the restricted isometry constant of our framework $\A=\U\D\tB$ can be written as
\begin{align}
\ric=\sup_{\x\in\sset}\left|\|\A\x\|_2^2-\|\x\|_2^2\right|, \label{eqn: ric dfn}
\end{align}
To complete the proof, our objective is to show that $\ric\leq\delta$ for $\delta\in(0,1)$ under the conditions in Theorem \ref{thm: the main theorem}.

We require the following important result due to Krahmer et al.:
\begin{theorem}[\cite{krahmer2014suprema}, Theorem 3.1]\label{thm: bounds on chaos processes}
Let $\mset$ be a set of matrices, and let $\bxi$ be a random vector whose entries $\xi_j$ are independent, mean $0$, variance $1$, and $r$-subgaussian random variables. Set
\begin{align*}
d_F(\mset)&=\sup_{\S\in\mset}\|\S\|_{F},\\
d_{\ttwo}(\mset)&=\sup_{\S\in\mset}\|\S\|_{\ttwo}.
\end{align*}
and
\begin{align}
&C_{\mset}(\bxi):=\sup_{\S\in\mset}\left|\|\S\bxi\|_2^2-\E\{\|\S\bxi\|_2^2\}\right|, \label{eqn: suprema of chaos}\\
&E=\gamma_2(\mset,\|\cdot\|_{\ttwo})\left[\gamma_2(\mset,\|\cdot\|_{\ttwo})+d_F(\mset)\right]\nonumber\\
& \quad +d_F(\mset)d_{\ttwo}(\mset), \nonumber\\
&V=d_{\ttwo}(\mset)\left[\gamma_2(\mset,\|\cdot\|_{\ttwo})+d_F(\mset)\right], \nonumber\\
&U=d_{\ttwo}^2(\mset). \nonumber
\end{align}
Then, for $t>0$,
\begin{align}\label{eqn: suprema of chaos prob bound}
\Prob(C_{\mset}(\bxi)\geq c_1E+t)\leq2\exp (-c_2\min\{\frac{t^2}{V^2},\frac{t}{U}\}).
\end{align}
The constants $c_1$, $c_2$ depends only on $L$.
\end{theorem}
Here, $C_{\mset}(\bxi)$ represents the suprema of chaos processes associated with a set of matrices $\mset$. This theorem implies that $C_{\mset}(\bxi)$ can be bounded by three parameters: the suprema of Frobenius norms $d_F(\mset)$, the suprema of operator norms $d_{\ttwo}(\mset)$ and a $\gamma_2$-functional $\gamma_2(\mset,\|\cdot\|_{\ttwo})$.

Without going into the details, we note that the $\gamma_2$-functional $\gamma_2(\mset,\|\cdot\|_{\ttwo})$ can be bounded in terms of the covering numbers $N(\mset,\|\cdot\|_{\ttwo},u)$ as below.
\begin{align}
\gamma_2(\mset,\|\cdot\|_{\ttwo})\leq c\int_{0}^{d_{\ttwo}(\mset)}\sqrt{\log N(\mset,\|\cdot\|_{\ttwo},u)} \mathrm{d}u, \label{eqn: gamma2 bound}
\end{align}
where the integral is known as Dudley integral or entropy integral \cite{talagrand2005generic}.

We now proceed to express the restricted isometry constant of our framework in such a form that its bound can be derived by appealing to Theorem \ref{thm: bounds on chaos processes}. Recall that $\A=\U\D\tB$ and $\D=\diag(\bxi)$. Let $\Vx=\U\diag(\tB\x)$, then the restricted isometry constant of $\A$ is
\begin{align*}
\ric&=\sup_{\x\in\sset}\left|\|\A\x\|_2^2-\xnorm^2\right| \\
    &=\sup_{\x\in\sset}\left|\|\Vx\bxi\|_2^2-\xnorm^2\right|,
\end{align*}
and clearly
\begin{align*}
\E\{\|\Vx\bxi\|_2^2\}=\E\{\|\A\x\|_2^2\}=\xnorm^2.
\end{align*}
Hence, the restricted isometry constant $\ric$ can be expressed as
\begin{align}\label{eqn: ric intermediate represetation}
\ric=\sup_{\x\in\sset}\left|\|\V_{\x}\bxi\|_2^2-\E\{\|\Vx\bxi\|_2^2\}\right|.
\end{align}
For each vector $\x\in\sset$, there exists a corresponding matrix $\V_{\x}$. We define the set of matrices associated with all $s$-sparse signals $\x\in\sset$ as
\begin{align*}
\mset_{V}=\{\V_{\x}:\x\in\sset\}.
\end{align*}
Then the restricted isometry constant \eqref{eqn: ric intermediate represetation} can be written as
\begin{align}\label{eqn: ric in suprema of chaos}
\ric=\sup_{\V_{\x}\in\mset_{V}}\left|\|\V_{\x}\bxi\|_2^2-\E\{\|\Vx\bxi\|_2^2\}\right|.
\end{align}
Therefore, we have completely express the restricted isometry constant of our framework in the form of Theorem \ref{thm: bounds on chaos processes} (by comparing \eqref{eqn: suprema of chaos} and \eqref{eqn: ric in suprema of chaos}), where $\S$ and $\mset$ are replaced with $\V_{\x}$ and $\mset_{V}$, respectively.

Now, before bounding the restricted isometry constant $\ric$ by using Theorem \ref{thm: bounds on chaos processes} \eqref{eqn: suprema of chaos prob bound}, we only need to estimate the three associated parameters $d_F(\mset_V)$, $d_{\ttwo}(\mset_V)$ and $\gamma_2(\mset_V,\|\cdot\|_{\ttwo})$.

By Proposition \ref{prop: UTF} and $\tB^*\tB=\I$,
\begin{align*}
\|\Vx\|_F=\|\tB\x\|_2=\xnorm=1, \forall \x\in\sset,
\end{align*}
which means
\begin{align}
d_F(\mset_V)=1. \label{eqn: parameter 1 bound}
\end{align}
Since any induced norm is a sub-multiplicative matrix norm,
\begin{align}
\|\Vx\|_{\ttwo}&\leq\|\U\|_{\ttwo}\|\diag(\tB\x)\|_{\ttwo} \nonumber\\
               &\leq\sqrt{\frac{\tn}{m}}\|\tB\x\|_{\infty}, \label{eqn: Vx operator norm bound}
\end{align}
where the last step is due to Proposition \ref{prop: UTF}.

For any vector $\x\in\sset$, we denote by $\x^s$ the length-$s$ vector that retains only the non-zero elements in $\x$. And correspondingly for any $\z\in\C^n$, we denote by $\z^s$ the length-$s$ vector that retains only the elements that have the same indexes as those of the non-zero elements in $\x$. Thus, we immediately have $\|\x^s\|_2=\xnorm=1$, and $\langle\z,\x\rangle=\langle\z^s,\x^s\rangle$.
Hence, for any $\x\in\sset$
\begin{align*}
\|\tB\x\|_{\infty} &= \max_{j\in[n]}\left\{|\langle \tB_{(j,:)},\x\rangle|\right\} \\
                   &= \max_{j\in[n]}\left\{|\langle \tB_{(j,:)}^s,\x^s\rangle|\right\} \\
                   &\leq \max_{j\in[n]}\left\{\|\tB_{(j,:)}^s\|_2\|\x^s\|_2\right\} \\
                   &= \|\x\|_2\max_{j\in[n]}\left\{\|\tB_{(j,:)}^s\|_2\right\} \\
                   &\leq \sqrt{s}\mu(\tB)\|\x\|_2 \\
                   &= \sqrt{s}\mu(\tB),
\end{align*}
where the first inequality is due to Cauchy-Schwarz inequality, and the second one is due to the definition of coherence.

Thus, we have
\begin{align}
d_{\ttwo}(\mset_V)=\sqrt{\frac{s\tn}{m}}\mu(\tB). \label{eqn: parameter 2 bound}
\end{align}
By \eqref{eqn: Vx operator norm bound}, it follows that for any two elements $\Vx,\Vy\in\mset_V$
\begin{align*}
\|\Vx-\Vy\|_{\ttwo}=\|\V_{\x-\y}\|_{\ttwo}\leq\sqrt{\frac{\tn}{m}}\|\x-\y\|_{\tinf}.
\end{align*}
Thus, for every $u>0$ we have
\begin{align*}
N(\mset_V,\|\cdot\|_{\ttwo},u)\leq N(\sset,\sqrt{\frac{\tn}{m}}\|\cdot\|_{\tinf},u).
\end{align*}

By Lemma \ref{lem: covering number estimates} and the fact that $N(\sset,\sqrt{\frac{\tn}{m}}\|\cdot\|_{\tinf},u)=N(\sset,\|\cdot\|_{\tinf},\sqrt{\frac{m}{\tn}}u)$, $N(\sset,\sqrt{\frac{\tn}{m}}\|\cdot\|_{\tinf},u)$ satisfies the following two bounds
\begin{align*}
& \sqrt{\log N(\sset,\sqrt{\frac{\tn}{m}}\|\cdot\|_{\tinf},u)}\lesssim \mu(\tB)\sqrt{\frac{s\tn}{mu^2}}\log\hn \\
&\sqrt{\log N(\sset,\sqrt{\frac{\tn}{m}}\|\cdot\|_{\tinf},u)} \\
& \quad \leq \sqrt{2s}\sqrt{\log(\frac{en}{s})+\log(1+\frac{2\mu(\tB)\sqrt{s}}{u}\sqrt{\frac{\tn}{m}})}
\end{align*}
We combine these inequalities to estimate the entropy integral \eqref{eqn: gamma2 bound}: we apply the first bound for $c\sqrt{\frac{\tn}{m}}\mu(\tB)\leq u\leq d_{\ttwo}(\mset_V)$, and the second bound for $0<u\leq c\sqrt{\frac{\tn}{m}}\mu(\tB)$, where $d_{\ttwo}(\mset_V)=\sqrt{\frac{s\tn}{m}}\mu(\tB)$, $c>0$ and $c^2\leq s$. It reveals that
\begin{align}
\gamma_2(\mset_V,\|\cdot\|_{\ttwo})\lesssim \sqrt{\frac{s\tn}{m}}\mu(\tB)(\log s)(\log \hn). \label{eqn: entropy integral for gamma2}
\end{align}
With the bounds for the three parameters $d_F(\mset_V)$, $d_{\ttwo}(\mset_V)$ and $\gamma_2(\mset_V
,\|\cdot\|_{\ttwo})$, the proof is completed by using Theorem \ref{thm: bounds on chaos processes}. Detail steps on the application of Theorem \ref{thm: bounds on chaos processes} follow those in Section 4 of \cite{krahmer2014suprema}.
\end{proof}
\section{Coherence Analysis}\label{app: coherence analysis}
Consider an $n\times n$ ($n=2^d$) unitary matrix $\B=\F\LLamda\W^*$, where $\F$ is the normalized DFT matrix, $\LLamda$ is a diagonal matrix whose diagonal entries are a Golay sequence constructed by the Rudin-Shapiro iterative process and $\W^*$ corresponds to the transpose of the orthonormal Haar matrix, which is defined by
\begin{align*}
\W_1^*&=[1],\\
\W_{2^l}^*&=\frac{1}{\sqrt{2}}\begin{bmatrix}\W_{2^l-1}^*\otimes\begin{bmatrix}1\\1\end{bmatrix}&\I_{2^l-1}\otimes\begin{bmatrix}1\\-1\end{bmatrix}\end{bmatrix},
\end{align*}
where $\otimes$ denotes the operation of kronecker product. We need to prove that $\mu(\B)\leq \sqrt{\frac{2}{n}}$.
\begin{proof}
Let $\one_q$ denote a column vector with $q$ ones. Define $\ggamma_{2q}$ as a length-$2q$ sequence by
\begin{align*}
\ggamma_{2q}=\begin{bmatrix}\one_q\\-\one_q\end{bmatrix}.
\end{align*}
Note that the columns of $\W^*$ can be written as
\begin{align*}
&\W_{(:,0)}^*=\frac{1}{\sqrt{2^d}}\one_{2^d},\\
&\W_{(:,1)}^*=\frac{1}{\sqrt{2^d}}\ggamma_{2^d}\\
&\W_{(:,2^l+s)}^*=\frac{1}{\sqrt{2^{d-l}}}\begin{bmatrix}\zero_{s\cdot2^{d-l}}\\ \ggamma_{2^{d-l}} \\ \zero_{(2^l-s-1)\cdot 2^{d-l}}\end{bmatrix},
\end{align*}
where $1\leq l\leq d-1$, $0\leq s\leq2^l-1$. Let $\llamda$ represent a length-$2^d$ Golay sequence constructed from the Rudin-Shapiro recursive process. Suppose we divide this sequence into $L=2^l$ segments $\llamda_{i}$ ($i=0,\cdots,2^l-1$), each of which is of length $2^{d-l}$ ($1\leq l\leq d-1$), i.e.,
\begin{align*}
\llamda=\begin{bmatrix}\llamda_0\\\llamda_1\\\vdots\\\llamda_{2^l-1} \end{bmatrix}.
\end{align*}
By definition of the recursive process, each segment $\llamda_{i}^T$ is still a Golay sequence. And it is also clear that $\llamda_{i}^T\odot\ggamma_{2^{d-l}}$ is a Golay sequence, where $\odot$ represents the element-wise multiplication. Therefore, we can easily get the bound $\mu(\B)\leq \sqrt{\frac{2}{n}}$.
\end{proof}

\bibliographystyle{IEEEtran}
\bibliography{bibfile}
\end{document}